\documentclass[submission,copyright,creativecommons]{eptcs}
\providecommand{\event}{ICLP 2019} 
\usepackage{breakurl}             
\usepackage{underscore}           
\usepackage{mathptmx}

\usepackage{booktabs, listings, xspace, mdframed, url} 
\usepackage{alltt}
\usepackage{xspace}
\usepackage{amsmath,amssymb}
\usepackage{enumitem}
\setlist{leftmargin=3ex, itemsep=.6ex}

\let\oldc\c
\newcommand{\Hex}[1]{\hspace{#1ex}}
\newcommand{\Vex}[1]{\vspace{#1ex}}
\newcommand{\mysec}[1]{\Vex{-2.5}\section{#1}}
\newcommand{\mysubsec}[1]{\Vex{-1.5}\subsection{#1\Vex{-.5}}}
\newcommand{\mypar}[1]{\Vex{-2}\paragraph{\bf #1.~}}
\newenvironment{example}{\Vex{-0.5}\mypar{\textit{\Hex{2.5}Example}}}{\hfill$\blacksquare$}
\newenvironment{code}{\Vex{0}\begin{alltt}\small}{\end{alltt}\Vex{0}}
\newcommand\co[1]{\texttt{\small #1}} 
\newenvironment{acks}{\Vex{-0.5}\mypar{\textit{Acknowledgments}}}

\newcommand\TD{\textsc{Top-down}\xspace}
\newcommand\BU{\textsc{Bottom-up}\xspace}
\newcommand\EBU{extended \BU} 

\renewcommand\rm[1]{\textrm{#1}}
\newcommand\defn[1]{\emph{#1}}
\newcommand\com[1]{}
\newcommand\m[1]{\mbox{$#1$}} 
\newcommand\IF{\m{\leftarrow}\xspace}
\newcommand\sinO[1]{\m{O(#1)}} 
\newcommand\sincO[1]{\sinO{\co{#1}}}  
\newcommand\sincsO[1]{\sincO{\##1}} 
\newcommand\doubcsO[2]{\sincO{\m{\co{\##1}\times\co{\##2}}}} 
\newcommand\ofour[4]{\m{O(\co{min}(\co{\##1}\times\co{\##2},\ \co{\##3}\times\co{\##4}))}}

\newtheorem{theorem}{{\bf Theorem}}
\newtheorem{lemma}{{\bf Lemma}}
\newtheorem{corollary}{{\bf Corollary}}
\newenvironment{proof}{\Vex{-.25}{\bf \textit{\Hex{2.5}Proof.~}}}{\hfill$\square$\Vex{1}}
\newenvironment{proofsketch}{\Vex{-.25}{\bf \textit{\Hex{2.5}Proof Sketch.~}}}{\hfill$\square$\Vex{1}}

\title{Extended Magic for Negation:\\Efficient Demand-Driven Evaluation of Stratified Datalog with Precise Complexity Guarantees}
\author{K. Tuncay Tekle \qquad Yanhong A. Liu
\institute{Department of Computer Science, Stony Brook University}
\email{tuncay,liu@cs.stonybrook.edu}
}
\def\titlerunning{Extended Magic for Negation}
\def\authorrunning{K. T. Tekle \& Y. A. Liu}

\begin{document}

\maketitle

\begin{abstract}
Given a set of Datalog rules, facts, and a query, answers
to the query can be inferred bottom-up starting from the
facts or top-down starting from the query. 
For efficiency, top-down evaluation is extended with memoization of inferred facts,
and bottom-up evaluation is performed
after transformations to make rules driven by the demand from the query.
Prior work has shown their precise complexity analysis and relationships.
However, when Datalog is extended with even stratified negation, 
which has a simple and universally accepted semantics,
transformations to make rules demand-driven may result in non-stratified negation, 
which has had many complex semantics and evaluation methods.

This paper presents (1) a simple extension to demand transformation,
a transformation to make rules demand-driven for Datalog without negation, to support stratified negation,
and (2) a simple extension to an optimal bottom-up evaluation method for Datalog with stratified negation, to handle non-stratified negation in the resulting rules.
We show that the method provides precise complexity guarantees.  It is also optimal in that only facts needed for top-down evaluation of the query are inferred and each firing of a rule to infer such a fact takes worst-case constant time.
We extend the precise relationship between top-down evaluation and demand-driven bottom-up evaluation to Datalog with stratified negation. Finally, we show experimental results for performance, as well as applications to previously challenging examples.
\end{abstract}

\maketitle

\mysec{Introduction}

Datalog~\cite{Maier:2018:DCH:3191315.3191317} is a logic language for deductive databases~\cite{DBLP:books/aw/AbiteboulHV95}, security~\cite{DBLP:conf/sp/DeTreville02}, networking~\cite{DBLP:journals/cacm/LooCGGHMRRS09}, semantic web~\cite{DBLP:journals/ws/CaliGL12}, and many other applications~\cite{WarLiu17AppLP-arxiv}.

Given a set of Datalog rules, facts, and a query, answers to the query can be inferred using bottom-up evaluation starting with the facts or top-down evaluation starting with the query. The dominant strategies for efficient evaluations are top-down evaluation with \defn{variant tabling}~\cite{DBLP:conf/iclp/TamakiS86} to memoize and reuse answers to subqueries, and bottom-up evaluation with the \defn{magic set transformation} (MST)~\cite{DBLP:conf/pods/BancilhonMSU86} so that evaluation of the transformed rules are driven by demand from the query. 

The performance of Datalog engines is difficult to understand due to the multitude of factors involved. For example, the performance of different tabling strategies vary drastically~\cite{DBLP:conf/iclp/RaoRR96}, and bottom-up evaluation after MST may be much slower than the bottom-up evaluation of the original rules~\cite{DBLP:conf/sigmod/SereniAM08}. Choosing the best evaluation method for a given set of rules requires precise complexity analysis of each evaluation method. Recently, \defn{demand transformation} (DT), resulting in simpler rules with better space complexity than MST has been introduced, and the time and space complexities of top-down evaluation with tabling and bottom-up evaluation after DT has been precisely established.

Stratified negation is a simple kind of negation with a simple and universally accepted semantics, and it is handled gracefully by top-down evaluation with tabling requiring almost no change. 
However, MST applied to rules with stratified negation is known to result in rules with non-stratified negation, which has many complex semantics, such as well-founded semantics~\cite{DBLP:journals/jacm/GelderRS91} and stable model semantics~\cite{DBLP:conf/iclp/GelfondL88}. Under such semantics, precise complexity analysis is difficult, and worst-case behavior is prohibitive, as in the quadratic time complexity in the size of the grounded rules for well-founded semantics.
A more recent semantics, founded semantics~\cite{DBLP:conf/lfcs/LiuS18}, can evaluate a certain class of rules with arbitrary negation in linear time, but not for all rules.
Other efforts to address this problem have been made, as discussed in Section~\ref{sec:related}, but they are more complicated than even MST, and none of them provide precise complexity guarantees, or proofs that they match top-down evaluation with tabling.

This paper presents a new method consisting of (1) a simple extension to demand transformation, a transformation to make rules demand-driven for Datalog without negation, to support stratified negation, and (2) a simple extension to an optimal bottom-up evaluation method for Datalog with stratified negation, to handle non-stratified negation in the resulting rules.  We show that the method provides precise complexity guarantees, which is worst-case linear time in the size of the grounded rules, in contrast to existing quadratic time results. The method only infers facts needed for top-down evaluation of the query, and each firing of a rule to infer such a fact takes worst-case constant time.  We extend the precise relationship between top-down evaluation and demand-driven bottom-up evaluation to Datalog with stratified negation.  Finally, we show experimental results for performance, as well as applications to previously challenging examples. 

\mysec{Datalog with negation} 

\defn{Datalog}~\cite{Maier:2018:DCH:3191315.3191317} is a language for defining rules, facts, and queries, where rules can be used with facts to answer queries.  

A Datalog rule is of the form:
\[
p(a_1,\rm{...},a_k)~ \IF~
p_1(a_{11},\rm{...},a_{1k_1}),\rm{...},\  
p_h(a_{h1},\rm{...},a_{hk_h}).
\]
where $h$ is a finite natural number, each $p_i$ (respectively $p$) is a predicate of finite number $k_i$ (respectively $k$) arguments, each $a_{ij}$ and $a_i$ is either a constant or a variable, and each variable in the arguments of $p$ must also be in the arguments of some $p_i$.

If $h=0$, then each $a_i$ must be a constant, in which case $p(a_1,\rm{...},a_k)$ is called a \defn{fact}.

For the rest of the paper, \defn{rule} refers only to the case where $h\geq 1$, in which case each $p_i(a_{i1},\rm{...},a_{ik_i})$
is called a \defn{hypothesis}, also known as \defn{premise}, and $p(a_1,\rm{...},a_k)$ is called the \defn{conclusion}.

The meaning of a set of rules and facts is the set of facts that are given or can be inferred using the rules.

A \defn{query} is of the form \m{p(a_1,\rm{...},a_k)?} where each argument \m{a_i} is either a constant or a variable. Answers to a query is the set of facts, given or inferred, that match the query, where a fact $f$ \defn{matches} a query $q$ if there is a substitution of the variables in $q$ to constants such that $q$ under the substitution is the same as $f$.

We call a predicate that appears in the conclusion of a rule an \defn{intensional} predicate, and a predicate that does not appear in the conclusion of any rule but appears only in given facts an \defn{extensional} predicate. We assume without loss of generality that intensional and extensional predicates are distinct.

In examples, we use letters for variables in arguments of predicates, and we use literal numbers for constants.

\begin{example} 
We use the following rules for transitive closure as a running example.

\begin{code}
p(x,y) \IF e(x,y).          \rm{(R1)}
p(x,z) \IF e(x,y), p(y,z).  \rm{(R2)}
\end{code}

We consider the query \co{p(1,x)?}.  It asks for all nodes reachable from \co{1} following the edges.
\end{example}

\mypar{Stratified negation}

Datalog can be extended with \defn{stratified negation} by extending rules to be of the form, where [\co{not}] indicates that \co{not} is optional:
\[
p(a_1,\rm{...},a_k)~ \IF~
[\co{not}] p_1(a_{11},\rm{...},a_{1k_1}),\rm{...},\  
[\co{not}] p_h(a_{h1},\rm{...},a_{hk_h}).
\]
and imposing the constraint that there is no cyclic dependency between any predicate and a negated predicate, i.e., a predicate preceded with a \co{not}. For example, replacing (R2) in the running example with the following rule, where \co{q} is an extensional predicate, results in Datalog with stratified negation:

\begin{code}
p(x,z) \IF e(x,y), p(y,z), \co{not} q(x,z).  \rm{(R2'S)}
\end{code}

\noindent because there is no cyclic dependency between \co{p} and \co{q}.

Formally, we say that predicate $p_1$ (negatively) depends on predicate $p_2$ if there is a rule whose conclusion's predicate is $p_1$ and there is a (negative) hypothesis in the rule whose predicate is $p_2$. A set of rules is \defn{stratified} if for any two predicates $p_1$ and $p_2$, if $p_1$ negatively depends on $p_2$, then $p_2$ does not transitively depend on $p_1$. Predicates in a stratified set of rules can be split into numbered \defn{strata} such that for any two strata $s_1$ and $s_2$, if a predicate in $s_1$ negatively depends on a predicate in $s_2$, then $s_1$ has a higher number than $s_2$.

Datalog with stratified negation has a simple and universally accepted semantics, where rules having predicates in a lower strata are evaluated first, and negated hypotheses under a substitution is considered false if the corresponding facts have not been inferred. 
This semantics coincides with all well-known semantics for negation, including perfect model semantics~\cite{DBLP:journals/jar/Przymusinski89}, well-founded semantics~\cite{DBLP:journals/jacm/GelderRS91}, and stable model semantics~\cite{DBLP:conf/iclp/GelfondL88}. 

\mypar{Non-stratified negation}

Removing the constraint that there is no cyclic dependency between any predicate and a negated predicate may result in non-stratified negation, which has no universally accepted semantics. Consider the following rule:

\begin{code}
t(x) \IF \co{not} t(x).
\end{code}

\noindent For a constant \co{c}, if \co{t(c)} is false, then it is also true by this rule, therefore standard semantics does not apply.
Many sophisticated semantics have been developed, including well-founded semantics~\cite{DBLP:journals/jacm/GelderRS91} and stable model semantics~\cite{DBLP:conf/iclp/GelfondL88}.  
Also, under sophisticated semantics, precise complexity analysis is difficult, and worst-case behavior is prohibitive, e.g., quadratic time in the size of the grounded rules for well-founded semantics~\cite{berman1995computing}.
Therefore, rules with non-stratified negation must be avoided whenever possible.

\mysec{Demand-driven evaluation and challenge of stratified negation}

This section describes top-down and bottom-up methods for answering queries,
and challenge of stratified negation.

\mypar{Top-down evaluation with tabling} 

To answer a query, top-down evaluation starts with the query, generates subqueries from hypotheses of rules whose conclusions match the query, considering rules in the order given, and considering hypotheses from left to right, and does so repeatedly until the subqueries match given facts.

Straightforward top-down evaluation may lead to repeated subqueries, or even infinite recursion for recursive rules such as (R2) in the running example. To address this problem, \defn{tabling} memoizes answers to queries encountered, and reuses the answers when a query is encountered again.

In this paper, we consider top-down evaluation using the dominant tabling strategy, variant tabling without early completion, which exploits all ways to infer the answers to a query modulo variable renaming~\cite{DBLP:journals/jacm/ChenW96}.  We refer to this evaluation as \TD in the rest of the paper. 

Stratified negation is generally supported
for queries that are \defn{non-floundering}~\cite{DBLP:journals/jacm/ChenW96}.  A query is non-floundering with respect to a set of rules if during \TD of the query, all subqueries for negated hypotheses have all of their arguments bound.

In this paper, we only consider non-floundering queries of Datalog with stratified negation, because these queries can be evaluated by \TD with a trivial extension to test the negation.

\mypar{Optimal bottom-up evaluation and demand-driven transformations}

Bottom-up evaluation starts with given facts, infers new facts from conclusions of rules whose hypotheses match existing facts, and does so repeatedly until no more facts can be inferred.

An optimal method~\cite{DBLP:journals/toplas/LiuS09} transforms any given set of Datalog rules into an efficient specialized procedural implementation with guaranteed worst-case time and space
complexities, and computes the complexities from the rules.

In particular, rules are first transformed to remove singleton variables (variables that occur in only one hypothesis) and multiple occurrences of the same variable in a single hypothesis. Next, rules with more than two hypotheses are decomposed into rules of two hypotheses.
Then, the least fixed-point specification of the rules is transformed to a while-loop that considers given and inferred facts incrementally, one at a time, where auxiliary indices are used to find each matching fact in constant time.

The evaluation is optimal in that only combinations of facts that make the hypotheses of a rule simultaneously true are considered, and each such combination, which leads to a \defn{firing} of the rule, is considered once in constant time. The time complexity is precisely the sum of the number of firings over all rules plus the size of extensional predicates for reading given facts.

In this paper, we only consider the decomposition of rules that takes the two leftmost hypotheses into a new rule at a time.
We call this \defn{left-optimal bottom-up evaluation}, because the time complexity of evaluation using this decomposition is optimal for the left-to-right ordering of the hypotheses in a rule. We refer to this evaluation as \BU in the rest of the paper.

Note that the optimal bottom-up evaluation in~\cite{DBLP:journals/toplas/LiuS09} is not limited to the particular decomposition used by \BU.  It can perform time and space complexity calculation for all possible decompositions, and select an optimal one trading space for time, including space for intermediate predicates from the decompositions.  We consider \BU in this paper in order to establish correspondence with the left-to-right evaluation of hypotheses in \TD. 

\BU infers all facts possible from the given facts, and thus may infer many facts not needed for answering the query.  For efficient evaluation, the query can be used to limit the facts that can be inferred. This is achieved by transforming the rules to be restricted by the query, like in the well-known \defn{magic set transformation (MST)}~\cite{DBLP:conf/pods/BancilhonMSU86}. 

\defn{Demand transformation (DT)}~\cite{DBLP:conf/ppdp/TekleL10} is such a transformation that results in rules with the same time complexity and exponentially better space complexity in program size
than MST.  It transforms a set of rules and a query into a new set of rules and a fact, 
which can infer only facts that can be inferred during \TD of the original rules. DT and MST make use of predicate annotations for argument binding patterns. DT differs from MST in that given intensional predicates 
are not annotated in transformed rules, therefore is much simpler, and much less space is used because a fact of a predicate is stored only once, rather than once for each possible annotation of the predicate.

DT has two stages: (I) compute demand patterns for intensional predicates, 
that is, a set of pairs $\langle\co{p},\co{s}\rangle$ 
indicating that under \TD for the given query, there is a subquery for predicate \co{p} with pattern \co{s}, a string of characters for arguments of the subquery, where the $i$th character is \co{b} if the $i$th argument is a constant (i.e., \defn{bound}), 
and \co{f} if it is a variable (i.e., \defn{free}), and (II) transform given rules and query 
using the demand patterns computed, in three steps, as follows.
\newcommand\DTbox{
\Vex{-0.25}
\begin{mdframed} 
\begin{enumerate}

\item[1.] For each demand pattern $\langle \co{p}, \co{s}\rangle$, and each given rule \co{p(\textrm{}args) \IF~h$_1$,\textrm{...},h$_n$.}, 
generate following rule

\quad\quad \co{p(\textrm{}args) \IF~d\_p\_s(a$_1$,\textrm{...},a$_k$), h$_1$,\textrm{...},h$_n$.}

where \co{args} denotes arguments in the conclusion, \co{a$_1$},\textrm{...},\co{a$_k$} are arguments in \co{args} 
that are bound by $\co{s}$, i.e., that correspond to character \co{b} in $\co{s}$.

\item[2.] For the given query $q$ of the form \co{p(args)?}, generate the following fact:

\quad\quad \co{d\_p\_s(a$_{b1}$,\textrm{...},a$_{bl}$).}

where \co{s} is the pattern for $q$, and \co{a$_{b1}$},\textrm{...},\co{a$_{bl}$} are the constant arguments in \co{args}. 

\item[3.] For each rule generated in Step 1, \co{c \IF~h$_0$,\textrm{...},h$_n$.}, and each \co{h$_i$} whose predicate is an intensional predicate \co{p}, generate the following rule:

\quad\quad \co{d\_p\_s(a$_1$,\textrm{...},a$_k$) \IF~h$_0$,\textrm{...},h$_{i-1}$.}

where \co{s} is the pattern of \co{h$_i$}, and \co{a$_1$},\textrm{...},\co{a$_k$} are the bound arguments of \co{h$_i$}.

\end{enumerate}
\end{mdframed}
\Vex{-1}
}

\DTbox
\normalsize

\begin{example}
For the running example, the only demand pattern is $\langle\co{p},\co{bf}\rangle$, because the only type of query encountered during \TD is one where the first argument of \co{p} is a constant. DT yields the following rules and fact:

\begin{code}
p(x,y) \IF d_p_bf(x), e(x,y).          \rm{(R1')}   (DT Step 1)
p(x,z) \IF d_p_bf(x), e(x,y), p(y,z).  \rm{(R2')}   (DT Step 1)
d_p_bf(1).                            \rm{(DF)}   \,(DT Step 2)
d_p_bf(y) \IF d_p_bf(x), e(x,y).       \rm{(DR)}   \,(DT Step 3)
\end{code}
(R1') and (R2') restrict the original (R1) and (R2) using a new demand predicate \co{d\_p\_bf} indicating demand on the first of two arguments of predicate \co{p}. (DF) is a new fact that corresponds to demand by the given query where ther first of the two arguments of \co{p} is \co{1}. (DR) propagates demand in (R2') from \co{d\_p\_bf(x)} via \co{e(x,y)} to demand on the first argument \co{y} in \co{p(y,z)}.
\end{example}

\mypar{Challenge of stratified negation for  
MST and DT}

Not only can \TD handle non-floundering queries of stratified Datalog with a trivial extension~\cite{DBLP:journals/jacm/ChenW96}, which checks whether a negative hypothesis is true under the current substitution; 
but also can \BU be trivially extended and give precise time and space complexity guarantees as for without negation~\cite{DBLP:journals/toplas/LiuS09}, by checking whether a fully instantiated negative hypothesis is true in $O(1)$ time.

It is therefore natural to think that MST and DT would apply to stratified negation as well, adding demands that mimic \TD, so that 
\BU of the resulting rules has the same optimal complexity as for without negation. Unfortunately, this is not true---applying MST to stratified rules, treating demand for a negated hypotheses as demand for the hypothesis without the negation, may result in non-stratified rules~\cite{DBLP:journals/jlp/BalbinPRM91}.
We show that DT has the same problem as MST through the following example.

\begin{example}
Consider extending the running example with two additional rules:

\begin{code}
p2(x,y) \IF not p(x,y), e2(x,y).           \rm{(R3)}
p2(x,z) \IF not p(x,z), e2(x,y), p2(y,z).  \rm{(R4)}
\end{code}
\co{p2} is the transitive closure of \co{e2} whose computation does not use any pair in \co{p}.
The four rules are stratified because only \co{p2} depends negatively on \co{p}, and \co{p} does not transitively depend on \co{p2}.

Consider the query \co{p2(1,2)?}. The query is non-floundering because the only demand pattern for the negated hypotheses is $\langle\co{p},\co{bb}\rangle$.
DT yields the following rules, treating demand for a negated hypothesis as demand for the hypothesis without the negation:

\begin{code}
p(x,y) \IF d_p_bb(x,y), e(x,y).                          \rm{(R1')}   \,(DT Step 1) 
p(x,z) \IF d_p_bb(x,z), e(x,y), p(y,z).                  \rm{(R2')}   \,(DT Step 1)
p2(x,y) \IF d_p2_bb(x,y), not p(x,y), e2(x,y).           \rm{(R3')}   \,(DT Step 1)
p2(x,z) \IF d_p2_bb(x,z), not p(x,z), e2(x,y), p2(y,z).  \rm{(R4')}   \,(DT Step 1)
d_p2_bb(1,2).                                           \rm{(DF)}   \,\,(DT Step 2)
d_p_bb(y,z) \IF d_p_bb(x,z), e(x,y).                     \rm{(D1)}    (DT Step 3)
d_p_bb(x,y) \IF d_p2_bb(x,y).                            \rm{(D2)}    (DT Step 3)
d_p_bb(x,z) \IF d_p2_bb(x,z).                            \rm{(D3)}    (DT Step 3)
d_p2_bb(y,z) \IF d_p2_bb(x,z), not p(x,z), e2(x,y).      \rm{(D4)}    (DT Step 3)
\end{code}
The resulting rules are not stratified: \co{d\_p2\_bb} negatively depends on \co{p}, in (D4), and \co{p} transitively depends on \co{d\_p2\_bb} through \co{d\_p\_bb}, in (R1') and (D2).
\end{example}

How does \TD handle stratified negation without changes? The key is 
that \TD has a sequential order for when each subquery is processed, whereas 
\BU infers facts of demand predicates, i.e., the predicates with prefix \co{d\_}, without a sequential order.
We illustrate this for the extended example above.

\begin{example}
In the example above, \TD would ask \co{p2(1,2)?}, which asks \co{not p(1,2)?} using rule (R4), and assuming it returns true, \TD would find a value for \co{y}, say \co{3} and ask \co{p2(3,2)?}, then again using (R4) would ask \co{not p(3,2)?}, and so on. 

In contrast, in the resulting rules from DT, a demand for \co{p2} for the third hypothesis of (R4) depends on 
the truth value of \co{p}, whose demand depends on the demand for \co{p2} as it appears in a rule whose conclusion's predicate is \co{p2}, leading to a cyclic definition with negation. 
\end{example}

It has been shown that well-founded semantics of magic-set transformed stratified rules is two-valued~\cite{DBLP:journals/tcs/KempSS95}; it can be shown easily that this applies to DT as well. This agrees with the understanding that \TD computes well-founded semantics~\cite{DBLP:journals/jacm/ChenW96}. Then, one could  
evaluate such rules after DT by computing well-founded semantics.  However, the best bottom-up methods known for computing well-founded semantics
are quadratic time in the grounded rules~\cite{DBLP:journals/tplp/BrassDFZ01}, which implies a prohibitive $O(k^{v}\times k^v)$ bound, where $k$ is the number of constants, and $v$ is the maximum number of variables in a rule.

A recent semantics, founded semantics~\cite{DBLP:conf/lfcs/LiuS18} can compute well-founded semantics in linear time in the absence of what they call \defn{closed predicates}. However, it can be shown that the example above requires closed predicates. 

A number of serious efforts have been made to address the problem of non-stratified negation in the resulting rules from MST, as discussed in Section~\ref{sec:related}, but all of them are much more complicated than even MST, and none of them provide precise complexity guarantees, let alone matching \TD.

\mysec{Extending demand transformation and bottom-up evaluation for stratified negation}

We present simple extensions to DT and \BU so that the demand from the query is precisely captured, and answers to the query are inferred efficiently with precise time complexity guarantees.
The resulting evaluation and complexity match \TD exactly.

\mysubsec{Extended demand transformation} 

We extend DT to handle stratified negation by introducing a new predicate \co{n.p} for each predicate \co{p} that appears in a negated hypothesis and using \co{n.p} in place of \co{not p}.
Predicate \co{n.p} therefore stands for the complement of \co{p}.
We use the demand for \co{n.p} to create demand for \co{p}.  
We then extend \BU in the next subsection to infer facts of \co{n.p} by exploiting stratification. The extended DT has three steps.

\vspace{0ex}
\begin{mdframed}
\begin{enumerate}
    \item[1.] Replace each negated 
    hypothesis \co{not p(args)} with \co{n.p(args)}, where \co{n.p} is a new predicate.
    
    \item[2.] Add a rule \co{n.p(a$_1$,\rm{...},a$_k$) \IF not p(a$_1$,\rm{...},a$_k$).} for each \co{n.p} of \m{k} arguments, where a$_1$,\rm{...},a$_k$ are distinct variables. 

    \item[3.] Apply DT as before, treating demand for
    \co{not p(args)} as demand for \co{p(args)}.  

    Note that the only rules that contain such negated hypotheses are rules added in Step 2.
\end{enumerate}
\end{mdframed}
\vspace{0ex}

\begin{example} 
For the extended running example, extended DT yields the
following rules.

\begin{code} 
p(x,y) \IF d_p_bb(x,y), e(x,y).                       \rm{(R1')}   (Step 3,    DT Step 1)
p(x,z) \IF d_p_bb(x,z), e(x,y), p(y,z).               \rm{(R2')}   (Step 3,    DT Step 1)
p2(x,y) \IF d_p2_bb(x,y), n.p(x,y), e2(x,y).          \rm{(R3'')}  \,(Steps 1,3, DT Step 1)
p2(x,z) \IF d_p2_bb(x,z), n.p(x,z), e2(x,y), p2(y,z). \rm{(R4'')}  \,(Steps 1,3, DT Step 1)
n.p(x,y) \IF d_n.p_bb(x,y), not p(x,y).               \rm{(N1)}   \,\,(Steps 2,3, DT Step 1)
d_p2_bb(1,2).                                        \rm{(DF)}   \,(Step 3,    DT Step 2)
d_p_bb(y,z) \IF d_p_bb(x,z), e(x,y).                  \rm{(D1)}   \,\,(Step 3,    DT Step 3)
d_n.p_bb(x,y) \IF d_p2_bb(x,y).                       \rm{(D2')}   (Step 3,    DT Step 3)
d_n.p_bb(x,z) \IF d_p2_bb(x,z).                       \rm{(D3')}   (Step 3,    DT Step 3)
d_p2_bb(y,z) \IF d_p2_bb(x,z), n.p(x,z), e2(x,y).     \rm{(D4')}   (Steps 1,3, DT Step 3)
d_p_bb(x,z) \IF d_n.p_bb(x,z).                        \rm{(DN1)}  \,(Steps 2,3, DT Step 3)
\end{code}

Note the following observations: 
\begin{itemize}
    \item Except for (N1), 
    the resulting rules contain no negation, because \co{not p(args)} has been replaced with \co{n.p(args)} in all other rules.
    \item 
    The rule (DN1) is added in Step 3 of extended DT for the demand resulting from the second hypothesis of the rule (N1).
    \item The new predicate \co{n.p} is still in a cycle containing negation with \co{p}, so the rules are not stratified.
\end{itemize}\Vex{-4}
\end{example}

Next, we extend bottom-up evaluation to handle resulting rules that contain negation so that facts of new predicates \co{n.p}, and in turn facts of all predicates, can be inferred correctly. 

\mysubsec{Extended bottom-up evaluation}

We extend \BU so that facts for each predicate \co{n.p} are inferred and the rest of the evaluation proceeds as before with optimal time complexity. 

We make the following observation that underlies the idea of the extension.

\begin{lemma}\label{lem:firststep} Given a set of stratified rules $rs$ and a query, let $rs'$ be the resulting rules after extended DT. After \BU of $rs'$, consider predicate \co{p} in the lowest stratum in $rs$ among predicates
such that\linebreak \co{d\_n.p\_s(args)} has been inferred: if \co{p(args)} has not been inferred, then the generated rule \co{n.p(a$_1$,\rm{...},a$_k$) \IF d\_n.p\_s(a$_1$,\rm{...},a$_k$), not p(a$_1$,\rm{...},a$_k$).} by extended DT can be used to infer \co{n.p(args)}.
\end{lemma}
\begin{proof}
Since \co{p} is in the lowest stratum among all predicates for which there is a negated demand, there is no negated hypothesis that \co{p} transitively depends on. If \co{d\_n.p\_s(args)} has been inferred, then because of a rule introduced in Step 3 of extended DT, \co{d\_p\_s(args)} has been inferred as well. Since \co{p} only depends on \co{d\_p\_s} and other positive hypotheses, \co{p(args)} must be inferred by \BU if true, and not inferred otherwise. Therefore if it has not been inferred, it is false. Since \co{n.p} is the complement of \co{p}, \co{n.p(args)} is true.
\end{proof}

Therefore, running \BU without rules containing negation, after \BU reaches a fixed point, new facts for a predicate \co{n.p} in the lowest stratum, for which the corresponding demand predicate \co{d\_n.p\_s} has a fact,
can be inferred. After inferring facts for \co{n.p}, \BU may be applied again, and reach a new fixed point. After the new fixed point is reached, if there are new facts for a demand predicate \co{d\_n.p\_s}, then \co{n.p} facts can be inferred again for such facts for the predicate in the lowest stratum, and so on. 
Precisely, \BU after extended DT is extended as follows, and we refer to it as \EBU after extended DT.

\begin{mdframed}
Repeat the following steps until fixed point.
\begin{enumerate}
    \item[1.] Perform \BU without using rules that infer facts of predicates of the form \co{n.p}.
    Such rules are first added in Step 2 of extended DT and then transformed in Step 3 of extended DT, and are of the form  \co{n.p(a$_1$,\rm{...},a$_k$) \IF d_n.p_s(a$_1$,\rm{...},a$_k$), not p(a$_1$,\rm{...},a$_k$).} where \co{s} is all $b$'s.
    \item[2.] Use a rule to infer a new fact \co{n.p(args)} if \co{d_n.p_s} is in the lowest stratum among predicates of the form \co{d_n.q_s} for some predicate \co{q} for which (i) \co{d_n.p_s(args)} has been inferred and (ii) \co{p(args)} has not been inferred.
\end{enumerate}
\end{mdframed}

We extend Lemma~\ref{lem:firststep} to show correctness at any iteration of the extended evaluation. 

\begin{lemma}\label{lem:allsteps} Suppose \co{d\_n.p\_s(args)} is inferred during \EBU after extended DT. If\linebreak \co{p(args)} is false for the original set of rules, then \co{n.p(args)} is inferred by \EBU after extended DT, and not inferred otherwise.
\end{lemma}
\begin{proof}
(Case 1) \co{p(args)} is false for the original set of rules. Then, this fact cannot be inferred by the rules defining \co{p} in the rules obtained after extended DT since they are more restrictive than the original rules. Given that \co{d\_n.p\_s(args)} (say $f$) is inferred during \EBU after extended DT, it is guaranteed that if the rule for \co{n.p} is ever used, then \co{n.p(args)} will be inferred. The rule for \co{n.p} is guaranteed to be used since at some point in the iteration $f$ will be the fact in the lowest stratum in Step~2. 

(Case 2) \co{p(args)} is true for the original set of rules. Since \co{d\_n.p\_s(args)} is inferred during \EBU after extended DT, then \co{d\_p\_s(args)} is inferred also due to the rule generated in Step 3 of extended DT. Hence, \co{p(args)} will be inferred during \EBU after extended DT since there is a demand for it, and all of the hypotheses that define it must be in a lower stratum, and hence correctly inferred also. Since \co{p(args)} is inferred, \co{n.p(args)} cannot be inferred.
\end{proof}

\begin{theorem}[Extended \BU after extended DT matches \TD] 
A fact is inferred during \EBU of a set of stratified rules and query after extended DT iff it is inferred during \TD of the rules and query.
\end{theorem}
\Vex{-0.5}
\begin{proofsketch}
This relationship is already known for \BU after DT and \TD for rules without stratified negation, therefore we only need to extend the argument to handle negated hypotheses. 

If a demand fact for a negated hypothesis is inferred during \EBU after extended DT, then all hypotheses to its left must be true, therefore a corresponding query must be evaluated during \TD, and vice versa.  So demand facts inferred correspond exactly to subqueries encountered during \TD. By Lemma~\ref{lem:allsteps}, if there is a demand for a negated hypothesis, then its corresponding fact is inferred correctly. Therefore, \EBU after extended DT corresponds exactly to \TD in terms of inferred facts.
\end{proofsketch}

\mysec{Precise complexity guarantees}

We give precise time complexity analysis using the following parameters as in~\cite{DBLP:journals/toplas/LiuS09}.

\begin{itemize}

\item \co{\#p}: number of facts 
of predicate \co{p}, called \defn{size of} \co{p}.

\item \co{\#p.i$_1$,\rm{...},i$_n$/j$_1$,\rm{...},j$_m$}: maximum number of combinations of values taken by the \co{i$_1$},\rm{...},\co{i$_n$}-th arguments of facts 
of predicate \co{p}, given any fixed value of the \co{j$_1$},\rm{...},\co{j$_m$}-th arguments.

\end{itemize}

There are two forms of rules after decomposition of rules.  When a rule has one hypothesis, it is of the form: \co{p(x$_1$,\rm{...},x$_n$) \IF q(y$_1$,\rm{...},y$_m$).} 
The number of times this rule can fire is the number of facts of \co{q}, therefore the time complexity incurred by this rule is \sincsO{q}.
In fact, we can omit the complexity of such rules, because (i) if \co{q} is an extensional predicate, 
then all its facts need to be read in, therefore the complexity is already incurred by the reading of the input, 
(ii) if \co{q} is an intensional predicate, then its size is bound by the complexity of the rules that infer its facts.

When a rule has two hypotheses, it is of the form: \co{p(args) \IF q(x$_1$,\rm{...},x$_n$), r(y$_1$,\rm{...},y$_m$).} To calculate the number of firings, we can first consider facts of \co{q} and find matching facts of \co{r} such that common variables in the two hypotheses take the same value. Therefore, only variables in the second hypothesis but not in the first can take different values for each fact of \co{q}.  Let $C_{12}$ be the set of indices $j$ in $1..m$ such that \co{y$_j$} is a common variable in the two hypotheses, then the complexity is bounded by \doubcsO{q}{r.$i \notin C_{12}/j \in C_{12}$}. 
Symmetrically, we can consider facts of \co{r} and find matching facts of \co{q}. Let $C_{21}$ be symmetrically the set of indices $j$ in $1..n$ such that \co{x$_j$} is a common variable, then the complexity is also bounded by \doubcsO{r}{q.$i \notin C_{21}/j \in C_{21}$}. Both are upper bounds, so the complexity is bounded by the minimum of the two: 
\[\ofour{q}{r.$i \notin C_{12}/j \in C_{12}$}{r}{q.$i \notin C_{21}/j \in C_{21}$}\]

\begin{example}
For the running example, rule (R1) incurs the time complexity \sincsO{e},
and rule (R2) incurs the time complexity \ofour{e}{p.2/1}{p}{e.1/2}.
\end{example}

\mysubsec{Complexity characteristics of \EBU after extended DT} 

\begin{theorem}[Extensions preserve complexity]\label{extendedcomplexity}
Extended \BU after extended DT preserves the complexity characteristics and  optimality of \BU after DT.
\end{theorem}
\begin{proof}
The only difference in \EBU is Step 2 where a demand fact for predicates of the form \co{n.p} in the lowest stratum is found. Such a fact can be found in constant time in data complexity, by searching for demand facts starting from the lowest stratum upwards. Therefore, the extension preserves the complexity characteristics and the optimality of \BU after DT.
\end{proof}

The precise time and space complexity of \TD, and its relationship with \BU after MST and DT, were open until~\cite{DBLP:conf/ppdp/TekleL10}. The results in~\cite{DBLP:conf/ppdp/TekleL10} include the following theorems---they establish that the time complexity of \TD and that of \BU after DT are equal for rules with at most two hypotheses each, and for general Datalog rules,
\BU after DT is equal to or faster than \TD.

\begin{theorem}[\BU after DT equals \TD on decomposed rules~\cite{DBLP:conf/ppdp/TekleL10}]\Vex{-1}
  Let $P$ be a set of Datalog rules and a query, 
  such that the rules have no singleton variables, 
  and there are no more than two hypotheses per rule. 
  Let $P'$ be the set of rules and fact after demand transformation of $P$. 
  Let $T_{td}$ be the asymptotic time complexity of \TD of $P$, 
  and $T_{bu}$ be the asymptotic time complexity of \BU of $P'$.  
  Then, $T_{bu} = T_{td}$.
\end{theorem}

\Vex{-0.75}
\begin{theorem}[\BU after DT beats \TD~\cite{DBLP:conf/ppdp/TekleL10}]
  Let $P$ be a set of Datalog rules and a query.  
  Let $P'$ be the set of rules and fact after demand transformation of $P$.
  Let $T_{td}$ be the asymptotic time complexity of \TD of $P$, 
  and $T_{bu}$ be the asymptotic time complexity of \BU of $P'$. 
  Then, $T_{bu} \le T_{td}$.
\end{theorem}

By Theorem~\ref{extendedcomplexity}, and the fact that \TD behaves identically for stratified negation, these theorems can be extended as follows.

\begin{corollary}[Extended \BU after extended DT 
equals \TD on decomposed rules]\Vex{-1}~\linebreak
  Let $P$ be a set of stratified rules and a non-floundering query, such that the rules have no singleton variables, and there are no more than two hypotheses per rule. Let $P'$ be the set of rules and fact after extended demand transformation of $P$. Let $T_{td}$ be the asymptotic time complexity of \TD of $P$, and $T_{bu}$ be the asymptotic time complexity of \EBU of $P'$. Then, $T_{bu} = T_{td}$.
\end{corollary}

\Vex{-0.75}
\begin{corollary}[Extended \BU after extended DT beats \TD]
  Let $P$ be a set of stratified rules and a non-floundering query.  Let $P'$ be the set
  of rules and fact after extended demand transformation of $P$.
  Let $T_{td}$ be the asymptotic time complexity of \TD of $P$, and $T_{bu}$ be the asymptotic time complexity of \EBU
  of $P'$. Then, $T_{bu} \le T_{td}$.
\end{corollary}

\begin{example}\label{ex:complexity}
For the extended running example after extended DT, we calculate the time complexity by considering the most complex rule, copied below.

\begin{code}
p2(x,z) \IF d_p2_bb(x,z), n.p(x,z), e2(x,y), p2(y,z).  \rm{(R4'')}
\end{code}
Extended \BU decomposes this rule into three rules:

\begin{code}
i1(x,z) \IF d_p2_bb(x,z), n.p(x,z).  \rm{(R4''1)}
i2(x,y,z) \IF i1(x,z), e2(x,y).      \rm{(R4''2)}
p2(x,z) \IF i2(x,y,z), p2(y,z).      \rm{(R4''3)}
\end{code}
Rules (R4''1) and (R4''3) have a hypothesis containing all variables occurring in the rule, so they do not contribute any extra complexity. The time complexity incurred by (R4''2) is
\[\ofour{i1}{e2.2/1}{e2}{i1.1/2}.\] 
Note that this complexity is precise, and can be very small because \co{i1} is the intersection of \co{d\_p2\_bb} and \co{n.p}. 

The complexity benefit of \EBU after extended DT is even more apparent when contrasted to evaluating the resulting rules bottom-up for well-founded semantics whose worst-case time complexity is $O(k^6)$ where $k$ is the number of constants, whereas the worst-case time complexity using our method is $O(k^3)$.
\end{example}

\mysec{Applications and experiments}

\Vex{3}
\mysubsec{Experiments on the extended running example} 

To confirm the effectiveness of our method, we implemented our method to perform extended DT, and then created a Python program that performs \EBU on the resulting rules.
We compare our results with state-of-the-art systems:
clingo~\cite{DBLP:journals/aicom/GebserKKOSS11}, an ASP solver, DLV2, the latest version of a deductive database with ASP and dynamic magic sets~\cite{alviano2012magic, alviano2017asp}, and XSB~\cite{DBLP:journals/jacm/ChenW96}, the dominant
engine with 
tabling and well-founded semantics. 

We show experiments on the extended running example.
For clingo, which does not support answering queries, we give it the transformed set of rules and facts.
For XSB, we used a special construct, \co{load\_dynca} to turn off indexing when loading facts, improving the performance more than 10 times.
We use facts generated for predicates \co{e} and \co{e2}, by varying the number of nodes and edges. The running times
were captured on an Intel Core i5 2.8 GHz with 8 GB RAM, with PyPy 7.1.0 to run Python programs, clingo 5.3.0, DLV2, and XSB 3.8, 
including the time for reading facts and producing output, 
and are averaged over 5 runs. The results are shown in Table~\ref{table:exp}, where K stands for 1000.
\Vex{-2}
\begin{table}[h]\label{table:exp}
  \small
  \caption{Running times in secs comparing our extended \BU (EBU), clingo, DLV2, XSB}
  \Hex{4}
    \begin{tabular}{cr|r|rr|rr|rr}
      \hline
      \# Nodes& \# Edges& EBU & clingo & clingo\m{/}EBU & DLV2 & DLV2/EBU & XSB & XSB\m{/}EBU \\\hline
      1K & 200K &  2.392 & 5.549 & 2.31 & 1.336 & 0.56 & 1.201 & 0.50  \\
      1K & 400K &  4.597 & 9.845 & 2.14 & 2.599 & 0.57 & 2.471 & 0.54 \\
      1K & 600K &  6.898 & 13.783 & 1.99 & 4.972 & 0.72 & 3.606 & 0.52  \\ 
      2K & 600K &  7.871 & 18.129 & 2.30 & 5.573 & 0.71 & 3.676 & 0.47   \\
      2K & 800K &  10.999 & 24.059 & 2.18 & 8.012 & 0.73  & 4.731 & 0.43 \\
      2K & 1000K & 13.978 & 29.669 & 2.12 & 7.260 & 0.52 & 6.008 & 0.43\\
     \hline
    \end{tabular}
\end{table}
\Vex{-2}

One can see that the running times for \EBU are tightly coupled with the number of edges as shown in the formula in Example~\ref{ex:complexity}.
It is difficult to speculate on the time and space complexities for clingo and DLV2 because precise complexity analysis as shown for our method is not known. The strength of our method can be seen in contrast to these mature systems even 
though our Python program is not optimized for constant factors, and Python is known to be much slower than C and C++ used to implement XSB, DLV2, and Clingo.

\mysubsec{Balbin's not-reach-in-reach2 problem}

Balbin et al.~\cite[Example 12]{DBLP:journals/jlp/BalbinPRM91} give the following rules 
to show the challenge of non-stratified negation resulting from applying MST on stratified rules
(some predicates/variables have been renamed for readability):

\begin{code}
r(x) \IF s(x).                       \rm{(B1)}
r(x) \IF e(x,y), r(y).               \rm{(B2)}
r2(x) \IF s2(x).                     \rm{(B3)}
r2(x) \IF not r(x), e2(x,y), r2(y).  \rm{(B4)}
\end{code}
along with the query \co{r2(1)?}.  

Using our method, the resulting rules from extended DT are:

\begin{code}
r(x) \IF d_r_b(x), s(x).                      \rm{(B1')}
r(x) \IF d_r_b(x), e(x,y), r(y).              \rm{(B2')}
r2(x) \IF d_r2_b(x), s2(x).                   \rm{(B3')}
r2(x) \IF d_r2_b(x), n.r(x), e2(x,y), r2(y).  \rm{(B4')}
n.r(x) \IF d_n.r_b(x), not r(x).              \rm{(N1)}
d_r2_b(1).                                   \rm{(DF)}
d_r_b(y) \IF d_r_b(x), e(x,y).                \rm{(D1)}
d_n.r_b(y) \IF d_r2_b(x).                     \rm{(D2)}
d_r2_b(y) \IF d_r2_b(x), n.r(x), e2(x,y).     \rm{(D3)}
d_r_b(x) \IF d_n.r_b(x).                      \rm{(D4)}
\end{code}

These rules can be evaluated in linear time in the input size, because after decomposition into rules with two hypotheses, each rule has at least one hypothesis with extensional predicate \co{s}, \co{e}, \co{s2}, \co{e2}
that contains all variables in the rule.  Contrast this to bottom-up evaluation for well-founded semantics, which would have only the bound of $O(k^4)$, where $k$ is the number of constants, since there is a rule with two variables. 

\mysubsec{Meskes-Noack's no-extra-joins-in-path problem}

Meskes and Noack~\cite[Example 1]{DBLP:journals/ipl/MeskesN93} give the following rules
to show that a negated predicate does not need to be in a positive cycle for non-stratified rules to occur after MST
(some predicates/variables have been renamed for simplicity): 

\begin{code}
s(x) \IF q(x,z), r(z,y).              \rm{(M1)}
p(x,y) \IF e(x,y), not s(y).          \rm{(M2)}
p(x,z) \IF e(x,y), p(y,z), not s(y).  \rm{(M3)}
\end{code}
along with the query \co{p(1,y)?}. It can be seen that the only negated predicate is \co{s}, and it is simply defined using extensional predicates, but after DT, non-stratified rules are obtained. 

Extended DT yields the following rules:

\begin{code}
s(x) \IF d_s_b(x), q(x,z), r(z,y).             \rm{(M1')}
p(x,y) \IF d_p_bf(x), e(x,y), n.s(y).          \rm{(M2')}
p(x,z) \IF d_p_bf(x), e(x,y), p(y,z), n.s(y).  \rm{(M3')}
n.s(x) \IF d_n.s_b(x), not s(x).               \rm{(N1)}
d_p_bf(1).                                    \rm{(DF)}
d_n.s_b(y) \IF e(x,y).                         \rm{(D1)}
d_p_bf(y) \IF d_p_bf(x), e(x,y).               \rm{(D2)}
d_n.s_b(y) \IF d_p_bf(x), e(x,y), p(y,z).      \rm{(D3)}
d_s_b(x) \IF d_n.s_b(x).                       \rm{(D4)}
\end{code}
(M3') after decomposition gives the highest time complexity:
$\ofour{i}{p.2/1}{p}{i.1/2}$
where \co{i} is an intermediate predicate for the leftmost two hypothesis, and is a subset of \co{e} restricted by \co{d\_p\_bf} for the first argument. 
Like our original running example, 
the time complexity is bounded by $O(k^3)$ using \EBU after extended DT
versus $O(k^6)$ for bottom-up evaluation for well-founded semantics.

\mysec{Related work and conclusion}\label{sec:related}

Datalog has been extensively studied, especially methods for efficient evaluation~\cite{Maier:2018:DCH:3191315.3191317,DBLP:books/aw/AbiteboulHV95}.

Top-down evaluation with variant tabling was introduced in the 1980s~\cite{DBLP:conf/iclp/TamakiS86} and has been widely studied as well. It has been implemented in top-down evaluation engines such as XSB~\cite{DBLP:journals/jacm/ChenW96} and SWI~\cite{DBLP:journals/tplp/WielemakerSTL12}.
Optimal bottom-up evaluation was given in~\cite{DBLP:journals/toplas/LiuS09}. Incremental maintenance of Datalog with stratified negation was studied in~\cite{motik2019maintenance}.

Transformations for demand-driven bottom-up evaluation have been studied, in many forms of the well-known magic-set transformation (MST)~\cite{DBLP:conf/pods/BancilhonMSU86}, as well as the more recent demand transformation (DT)~\cite{DBLP:conf/ppdp/TekleL10} that our work extends. 

A method for precise calculation of time and space complexities of top-down evaluation with variant tabling for Datalog, as well as theorems establishing the precise relationship between top-down and bottom-up evaluations are given in~\cite{DBLP:conf/ppdp/TekleL10,DBLP:conf/sigmod/TekleL11}.

Although top-down evaluation handles stratified negation without changes, an analogous demand-driven bottom-up evaluation has been challenging to obtain. 

Balbin et al.~\cite{DBLP:journals/jlp/BalbinPRM91} introduce a transformation to a set of magic-set transformed rules, that works for an unidentified class of rules, and does not apply to all rules with stratified negation. They then introduce a new method for bottom-up evaluation with no complexity characterization. 

Meskes and Noack~\cite{DBLP:journals/ipl/MeskesN93} extend generalized supplementary MST~\cite{DBLP:journals/jlp/BeeriR91} by ignoring negated hypotheses for rules that infer demand facts to keep the rules stratified, but the method overestimates the demand and hence does not correspond to top-down evaluation.

Ross~\cite{DBLP:journals/jacm/Ross94} presents a new transformation whose output is not Datalog, but a higher-order logic, complicating their evaluation and giving no complexity analysis. 

Alviano et al.~\cite{alviano2012magic} extends magic sets for disjunctive Datalog, but its bottom-up computation requires grounding and stable model search, not giving precise complexity guarantees as we do.

It has been shown that the well-founded semantics of rules obtained by MST for stratified rules is two-valued~\cite{DBLP:journals/tcs/KempSS95}, but the best known complexity for bottom-up evaluation for well-founded semantics is a highly prohibitive quadratic time in the size of the grounded rules.

Our work extends demand transformation and optimal bottom-up evaluation and gives precise complexity guarantees for efficient demand-driven bottom-up evaluation of stratified rules, preserving optimality, and relationships with top-down evaluation. 

Future work includes methods for efficient bottom-up evaluation and demand-driven evaluation for non-stratified negation with precise complexity guarantees.

\begin{acks}
This work was supported in part by NSF under grants CCF-1414078 and IIS-1447549.
\end{acks}
\Vex{-3}

\renewcommand\c[1]{\oldc{#1}} 

\nocite{*}
\bibliographystyle{eptcs}
\bibliography{biblio}

\end{document}